\documentclass[]{article}
\usepackage[a4paper, total={6in, 8in}]{geometry}
\usepackage{authblk}

\usepackage{amsfonts}
\usepackage{amsmath}
\usepackage{amssymb}

\usepackage{mathtools} 

\usepackage[normalem]{ulem}

\usepackage{bm}
\usepackage{cancel}
\usepackage{todonotes}

\usepackage[citestyle=numeric,bibstyle=ieee,sorting=none,backend=biber,url=false,isbn=false]{biblatex}
\AtEveryBibitem{
	\clearfield{urlyear}
	\clearfield{urlmonth}
}
\usepackage{hyperref}
 
\usepackage{braket}

\usepackage[most]{tcolorbox}

\usepackage{amsthm}
\newtheorem{theorem}{Theorem}

\newtheorem{proposition}[theorem]{Proposition}
\newtheorem{corollary}[theorem]{Corollary}

\newtheorem{remark}[theorem]{Remark}

\theoremstyle{definition}
\newtheorem{example}[theorem]{Example}

\usepackage{environ}

\usepackage{tabularx}
\newcolumntype{C}{>{$}c<{$}}
\newcolumntype{L}{>{$}l<{$}}
\newcolumntype{R}{>{$}r<{$}}

\usepackage{multirow}

\renewcommand{\Re}{\mathrm{Re}}

\title{Perfect transmission and parallel composition for quantum walks on graphs with two leads}
\author[1]{Allan John Gerrard\footnote{gerrarda$\otimes$fc.jwu.ac.jp}}
\affil[1]{\it Department of Mathematics, Physics and Computer Science, Japan Women's University, 2--8--1 Mejirodai, Bunkyo-ku, Tokyo 112--8681, Japan}

\author[2]{Ryo Asaka}
\affil[2]{\it Department of Physics, Tokyo University of Science, Kagurazaka 1--3, Shinjuku-ku, 162--8601, Tokyo, Japan}

\author[2]{Kazumitsu Sakai}

\begin{document}

\maketitle

\begin{abstract}
    We study scattering for continuous-time quantum walks on finite graphs with
    two attached leads.  We derive explicit formulae for the two-terminal
    scattering matrix in terms of characteristic polynomials of the finite graph
    and its vertex-deleted subgraphs.  For real-weighted two-terminal 
    graphs, we then introduce three real
    quantities, $\mu_1$, $\mu_2$, and $\nu$,
    which are each additive under
    parallel composition of graphs.  In these variables, perfect 
    transmission at fixed momentum is characterized
    by the condition $\mu_1=\mu_2$ together with a hyperbola in the corresponding
    $(\mu,\nu)$-plane, whose points determine the transmission phase. 
    This turns the search for graphs with prescribed
    transmission properties into a geometric vector-sum problem for 
    smaller building blocks.
\end{abstract}

\section{Introduction}

    Continuous-time quantum walks on graphs with attached leads give a simple
    model of quantum scattering.  A finite graph is connected to semi-infinite
    paths, and an incoming particle is scattered by the finite part of the graph.
    The scattering matrix describes the reflection and transmission amplitudes.
    Childs and his collaborators \cite{childsUniversalComputationQuantum2009,childsUniversalComputationMultiparticle2013} showed that universal quantum computation can be
    implemented in this model by using finite graphs as scattering
    gadgets realizing quantum gates.
    More precisely, the transmission amplitudes of the corresponding
    scattering matrices give the matrix elements of the gates.  Thus, in
    gate applications,
    one is naturally led to search for graphs with perfect transmission at the
    operating momentum.

    This scattering viewpoint has also motivated possible physical realizations.
    For example, LC telegrapher circuits have been proposed as realizations of
    quantum-walk scattering gadgets for universal quantum gates
    \cite{ezawaElectricCircuitsUniversal2020}. Related quantum walks with tails, 
    in particular Szegedy walks, also have
    electric-circuit interpretations in terms of Kirchhoff-type laws and currents
    \cite{higuchiElectricCircuitInduced2020}.

    A difficulty is that the graph itself is the device.  Once the operating
    momentum is fixed, there is no continuously tunable external parameter.  The scattering
    amplitudes are determined by the structure of the finite graph.  Therefore, to construct a
    graph with prescribed scattering properties, one usually has to search through
    large families of graphs. Graphs which do not exhibit perfect transmission
    at the operating momentum are discarded.

    In this paper, we consider scattering on graphs with two terminals.  Although
    a general quantum gate usually requires more leads, two-terminal graphs
    already appear in phase gates
    \cite{childsUniversalComputationQuantum2009},
    momentum filters \cite{childsUniversalComputationQuantum2009}, and momentum
    switches \cite{childsMomentumSwitches2015}.  We derive formulae for the
    two-terminal scattering matrix in terms of characteristic polynomials of the adjacency matrix of the 
    finite graph and its vertex-deleted subgraphs.

    We then introduce three quantities $\mu_1$, $\mu_2$, and $\nu$, defined from
    these characteristic polynomials.  These quantities play a role similar to
    admittance in an electrical circuit.  In terms of these quantities, the condition for perfect
    transmission at a fixed momentum takes a simple form:
    $\mu_1=\mu_2$, together with a hyperbola in the corresponding
    $(\mu,\nu)$-plane.  The point on this hyperbola determines the
    phase of transmission.

    The main advantage of this parametrization appears when graphs are combined
    in parallel.  For parallel graphs, the quantities $\mu_1$, $\mu_2$, and
    $\nu$ are additive.  Hence, each small two-terminal graph gives a vector in
    the $(\mu_1,\mu_2,\nu)$-space, and parallel composition corresponds to vector
    addition.  Instead of searching directly through large graphs, one can then
    search for combinations of smaller building blocks whose vector sum satisfies
    the perfect-transmission condition.  In the symmetric case, this means that
    the sum satisfies $\mu_1=\mu_2$ and lies on the perfect-transmission hyperbola
    in the $(\mu,\nu)$-plane.  Thus, the graph search problem is reduced to a
    geometric vector-sum problem, giving a systematic way to construct graphs with
    prescribed transmission properties, such as momentum filters and phase gates.

    Characteristic polynomials are known to be effective tools for analysis of quantum walks, as evidenced by their use in quantum walks on finite graphs.  In that
    setting, vertex-deleted characteristic polynomials appear in exact formulae
    for spectral idempotents, which in turn control transition amplitudes.  This
    approach has been used to study state transfer, average mixing, and
    irrational quantum walks \cite{coutinhoIrrationalQuantumWalks2023}, as well as the effect
    of cut vertices and bridges on state transfer \cite{coutinhoQuantumWalksNot2022}.
    The present paper applies similar graph-polynomial data to a different
    problem, namely scattering on graphs with attached leads, and adapts it to
    the study of perfect transmission and parallel composition.
    As well as this, \cite{cottrellSimpleMethodFinding2015} makes use of characteristic polynomials to calculate $S$-matrix elements in the context of a quantum walk defined on edge states. 

    The paper is organized as follows.  In Section~2, we recall the scattering
    formalism for continuous-time quantum walks on graphs with leads and introduce
    the Schur-complement expression for the scattering matrix.  In Section~3, the
    characteristic-polynomial formulae for the two-terminal scattering matrix and
    the perfect-transmission condition are derived. In Section~4, the quantities
    $\mu_1$, $\mu_2$, and $\nu$ are introduced, the perfect-transmission
    hyperbola is described, and their additivity for parallel graphs is proved.
    The relation with admittance, poles, effective length, and series composition
    is also discussed.  In Section~5, we give some examples.  Section~6 contains
    the conclusion and discussion.
    Some supplementary material supporting the main text is deferred
    to the appendix.

\section{Preliminaries}
    Let $G$ denote a connected, edge-weighted graph with vertex set $V$ and weight function $\omega: V \times V \to \mathbb{C}$ satisfying $\omega(u,v)=\omega(v,u)^*$.
    Let $U \subset V$ denote a distinguished set of labelled vertices called \textit{terminals}, which are labelled $1,2, \dots, N$, with $N=|U|$.
    Let $\widehat{G}$ denote the infinite graph constructed by attaching semi-infinite paths to each of the terminals; these semi-infinite paths will be referred to as \textit{leads}.

    We consider a one-particle quantum mechanical system with a basis of position states corresponding to the vertices of graph $\widehat{G}$.
    In particular, denote by $\ket{\xi}$ the state where a particle is localised at $\xi \in V$, and similarly $\ket{j,x}$ for a particle at position $x$ along lead $j$, with $\ket{j,0}=\ket{j}$ referring to the terminal vertex of lead $j$ for $1 \leq j \leq N$.

    The Hamiltonian of the model, $H(\widehat{G})$, will simply be the (weighted) adjacency matrix of $\widehat{G}$. 
    Let $H(G)$ denote the adjacency matrix of $G$; that is, the matrix $(\omega(u,v))_{u,v \in V(G)}$, a finite Hermitian matrix.
    We write its block decomposition as
    \[
        H(G) = \begin{pmatrix}
            H(U) & B(G)^\dagger \\ B(G) & H(G \setminus U)
        \end{pmatrix}.
    \]

    Following \cite{childsLevinsonsTheoremGraphs2012}, let us define a scattering state $\ket{\psi_j(k)}$ with incoming particle along lead $j$ with momentum $k$ as the energy eigenstate of energy $\epsilon(k) = 2\cos(k)$ and
    \begin{equation} \label{smatrix-def}
        \braket{j',x |\psi_j(k)} = \delta_{j,j'} e^{-ikx} + [S(G,e^{ik})]_{j,j'}e^{ikx}.
    \end{equation}
    It was shown in \cite{childsLevinsonsTheoremGraphs2012} that
    \[
        [S(G,z)]_{j,j'} = -\left[\gamma(G,z^{-1})^{-1}\gamma(G,z)\right]_{j,j'},
    \]
    where
    \[
        \gamma(G,z) :=
        \begin{pmatrix}
            zI - H(U) & -B^\dagger \\ -B & (z+z^{-1})I-H(G\setminus U)
        \end{pmatrix}.
    \]

    We will use this to obtain a simpler expression of the $S$-matrix. 
    Let us apply the identity
    \[
        I-\gamma(G,z^{-1})^{-1}\gamma(G,z)
        =
        \gamma(G,z^{-1})^{-1}\left(\gamma(G,z^{-1})-\gamma(G,z) \right)
        =
        \gamma(G,z^{-1})^{-1}\begin{pmatrix}
            (z^{-1}-z)I & 0 \\ 0 & 0
        \end{pmatrix}.
    \]
    Looking at the left-hand side, the upper-left block of this expression is equal to $I+S(G,z)$.
    As such,
    \[
        [S(G,z)]_{j,j'}=-\delta_{j,j'}-(z-z^{-1})[\gamma(G,z^{-1})^{-1}]_{j,j'}.
    \]
    Next, we introduce the Schur complement of $\gamma(G,z)$ with respect to the $G\setminus U$,
    \begin{equation}\label{schur-comp}
            Q(G,z) := zI-H(U)-B^\dagger ((z+z^{-1})I-H(G\setminus U))^{-1} B,
    \end{equation}
    Then, using the Schur complement formula for the inverse of a block matrix, we obtain the following:
    \begin{equation} \label{smatrix-schur-comp}
        S(G,z) = - I -(z-z^{-1})Q(G,z^{-1})^{-1} .
    \end{equation}
    We will use this form of the $S$-matrix in later calculations.

    We also note the following transformation of the $S$-matrix. 
    The graph $G$ can be extended ``up lead $a$'' by redefining the location of the terminal to be $(a,1)$, and 
    shifting the index of vertices on lead $a$ by one: $(a,x) \mapsto (a,x-1)$.
    Let us denote the new graph by $G'$.
    The vertex labelled by $(a,0)$ in $G$ becomes an unlabelled vertex inside $G'$.
    This has no physical meaning, but changes the definition of the $S$-matrix.
    See Figure~\ref{fig:up-the-lead} for a visual description of this process.
    From the definition \eqref{smatrix-def}, it is clear that the effect on the $S$ matrix is the phase shift
    \begin{equation} \label{shift-terminals}
        [S(G,z)]_{j,j'} \mapsto [S(G',z)]_{j,j'} = z^{\delta_{aj}}z^{\delta_{aj'}}[S(G,z)]_{j,j'}.
    \end{equation}
    If the terminal $a$ is a leaf in the finite graph $G$ then the reverse operation is valid, and has the effect  $[S(G,z)]_{j,j'} \mapsto [S(G',z)]_{j,j'} = z^{-\delta_{aj}}z^{-\delta_{aj'}}[S(G,z)]_{j,j'}$.

\begin{figure}[ht]
    \centering
    \includegraphics{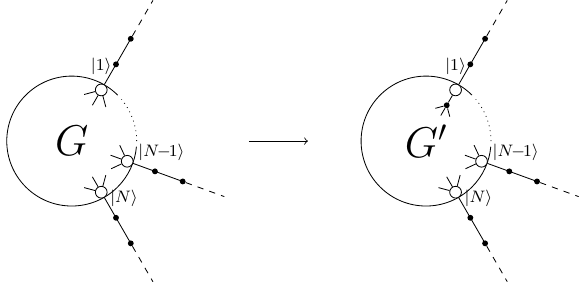}
    \caption{Visual description of the operation of extending the graph $G$ up the first lead.}
    \label{fig:up-the-lead}
\end{figure}

\section{The \texorpdfstring{$S$}{S}-matrix and characteristic polynomials}

    In what follows, it will be convenient to extend the graph $G$ up each lead by one position; we will denote this graph $G^{+1}$.
    In $G^{+1}$, the terminals are not connected, and are each adjacent to only one internal vertex of the graph (the terminal of the original $G$).
    Recall from \eqref{shift-terminals} that the $S$-matrices are related by
    \[
        [S(G^{+1},z)]_{j,j'} = z^2[S(G,z)]_{j,j'}.
    \]
    The effect of extending the graph from $G$ to $G^{+1}$ is to simplify the expression for $Q(G^{+1},z^{-1})$ \eqref{smatrix-schur-comp}.
    Indeed, since the terminal vertices of $G^{+1}$ are only adjacent to the terminal vertices of $G$, the matrix $B(G^{+1})^\dagger ((z+z^{-1})I-H(G))^{-1}B(G^{+1})$ picks out the elements of the resolvent matrix corresponding to the terminal vertices of the original $G$.
    Thus, the Schur complement formula \eqref{schur-comp} simplifies to
    \begin{equation} \label{Q-Gplus}
        Q(G^{+1},z^{-1})_{j,j'} = z^{-1}\delta_{j,j'} - \left(((z+z^{-1})I-H(G))^{-1}\right)_{j,j'}.
    \end{equation}

    The resolvent matrix appearing in \eqref{Q-Gplus} is also known as the walk-generating function, and has been extensively studied in e.g. \cite{godsilAlgebraicCombinatorics1993}.
    Introduce the characteristic polynomials of the adjacency matrix,
    \[
        \phi_{G}(y) = \det\left(y I - H(G)\right).
    \]
    The diagonal elements of the resolvent matrix admit the following expression:
    \begin{equation}
        \left((yI-H(G))^{-1}\right)_{j,j} = \frac{\det\left(yI-\widehat{H(G)}_{j,j}\right)}{\det\left(yI-H(G)\right)} = \frac{\phi_{G\setminus j}(y)}{\phi_{G}(y)},
    \end{equation}
    where $\phi_{G}(y)$ is the characteristic polynomial of the adjacency matrix of $G$ and $G\setminus j$ is the graph $G$ with the vertex $j$ and any edges attached to $j$ removed. 

    Further, from \cite{godsilToolsLinearAlgebra1989} (originally appearing in \cite{tutteAllKingsHorses}) we have the following identity for the off-diagonal elements:
    \begin{equation} \label{offdiag-formula}
        \left((yI-H(G))^{-1}\right)_{j,j'}\left((yI-H(G))^{-1}\right)_{j'j} = \frac{\phi_{G\setminus j}(y)\phi_{G\setminus j'}(y)-\phi_{G}(y)\phi_{G\setminus j\setminus j'}(y)}{\phi_{G}(y)^2}.
    \end{equation}
    We will only be interested in the case $y \in \mathbb{R}$.
    Since $H(G)$ is Hermitian, we have
    \begin{equation}
        \left|\left((yI-H(G))^{-1}\right)_{j,j'}\right|^2= \frac{\phi_{G\setminus j}(y)\phi_{G\setminus j'}(y)-\phi_{G}(y)\phi_{G\setminus j\setminus j'}(y)}{\phi_{G}(y)^2}.
    \end{equation}
    Thus, the value is determined up to a phase.
    For a real symmetric $H(G)$, i.e. when the graph is unweighted or has real weights,
    \begin{equation}
        \left((yI-H(G))^{-1}\right)_{j,j'}= \frac{\sqrt{\phi_{G\setminus j}(y)\phi_{G\setminus j'}(y)-\phi_{G}(y)\phi_{G\setminus j\setminus j'}(y)}}{\phi_{G}(y)} .
    \end{equation}
    As shown in \cite{godsilAlgebraicCombinatorics1993}, the sign of this square root can be discerned from an alternative form of the numerator, namely $\sum_{\mathcal{P}(G)_{j,j'}} \omega(P) \phi_{G\setminus P}(y)$, where $\mathcal{P}(G)_{j,j'}$ is the set of paths from $j$ to $j'$ in $G$, and $\omega(P)$ denotes the product of edge weights along edges in $P$. 
    As such, the sign of the square root is unambiguous. 
    
    Returning to $Q(G^{+1},z)$, we have 
    \begin{gather} \label{Qelements}
        \left[Q(G^{+1},z^{-1})\right]_{j,j} =
        z^{-1} - \frac{\phi_{G\setminus j}}{\phi_{G}}\bigg|_{z+z^{-1}};
        \\
        \left|\left[Q(G^{+1},z^{-1})\right]_{j,j'}\right| = 
        \frac{\sqrt{\phi_{G\setminus j}\phi_{G\setminus j'}-\phi_{G}\phi_{G\setminus j\setminus j'}}}{\phi_{G}} \bigg|_{z+z^{-1}}. \label{Qoffdiag}
    \end{gather}

\subsection{Two-terminal case}
    Returning to \eqref{smatrix-schur-comp}, in order to calculate the $S$-matrix, it remains to take the inverse of $Q(G^{+1},z^{-1})$.
    In the case of two terminals, the inverse matrix is simple to evaluate.
    Now, while general quantum gates require at least four terminals, two-terminal graphs appear in \cite{childsUniversalComputationQuantum2009} as momentum filters and momentum separators.

    \begin{proposition}
        For a graph $G$ with two terminals,
        \begin{align}
            \label{S11}
            S_{11}(G,z) &= -\frac{\phi_{G} - z^{-1}\phi_{G\setminus 1} - z\phi_{G\setminus 2}
            +\phi_{G\setminus 1\setminus 2} }{\phi_{G}
            - z\left(\phi_{G\setminus 1}+\phi_{G\setminus 2}\right)
            +z^2\phi_{G\setminus 1\setminus 2} }
            \\
            \label{S22}
            S_{22}(G,z) &= -\,\frac{\phi_{G} - z^{-1}\phi_{G\setminus 2}- z\phi_{G\setminus 1}
            +\phi_{G\setminus 1\setminus 2} }{\phi_{G}
            - z\left(\phi_{G\setminus 1}+\phi_{G\setminus 2}\right)
            +z^2\phi_{G\setminus 1\setminus 2} },
        \end{align}
        and the transmission probability is
        \begin{equation}
            \label{S12-prob}
            |S_{12}(G,z)|^2 = |S_{21}(G,z)|^2 = |z-z^{-1}|^2\frac{\phi_{G\setminus 1}\phi_{G\setminus 2}-\phi_{G}\phi_{G\setminus 1\setminus 2}}{|\phi_{G}
            - z\left(\phi_{G\setminus 1}+\phi_{G\setminus 2}\right)
            +z^2\phi_{G\setminus 1\setminus 2}|^2},
        \end{equation}
        where each characteristic polynomial is evaluated at $z+z^{-1}$.
        Further, if the weights of $G$ are all real, then 
        \begin{equation}
            \label{S12}
            S_{12}(G,z) = S_{21}(G,z) = (z-z^{-1})\frac{\sqrt{\phi_{G\setminus 1}\phi_{G\setminus 2}-\phi_{G}\phi_{G\setminus 1\setminus 2}}}{\phi_{G}
            - z\left(\phi_{G\setminus 1}+\phi_{G\setminus 2}\right)
            +z^2\phi_{G\setminus 1\setminus 2} }.
        \end{equation}
    \end{proposition}

    \begin{proof}
        From \eqref{Qelements} and \eqref{Qoffdiag}, we have
        \begin{align*}
            \det\left(Q(G^{+1},z^{-1})\right) &= \left(\left(z^{-1}\phi_{G} - \phi_{G\setminus 1}\right)
            \left(z^{-1}\phi_{G} - \phi_{G\setminus 2}\right)
            - \left(\phi_{G\setminus 1}\phi_{G\setminus 2}-\phi_{G}\phi_{G\setminus 1\setminus 2}\right)\right)/\phi_{G}^2
            \\
            &=\left(z^{-2}\phi_{G}^2
            - z^{-1}\phi_{G}\left(\phi_{G\setminus 1}+\phi_{G\setminus 2}\right)
            +\phi_{G}\phi_{G\setminus 1\setminus 2}\right)/\phi_{G}^2
            \\
            &=z^{-2}\left(\phi_{G}
            - z\left(\phi_{G\setminus 1}+\phi_{G\setminus 2}\right)
            +z^2\phi_{G\setminus 1\setminus 2} \right)/\phi_{G},
        \end{align*}
        where the functional dependence on $z+z^{-1}$ is suppressed.
        Then, from \eqref{smatrix-schur-comp},
        \[
            S_{12}(G^{+1},z) = (z-z^{-1}) \frac{[Q(G^{+1},z^{-1})]_{12}}{\det Q(G^{+1},z^{-1})}.
        \]
        This gives both the results for the off-diagonal elements. 
        For the diagonal elements,
        \begin{align*}
            S_{11}(G^{+1},z) &= -1-\frac{z^2(z-z^{-1})(z^{-1}\phi_{G} - \phi_{G\setminus 2})}{\phi_{G}
            - z\left(\phi_{G\setminus 1}+\phi_{G\setminus 2}\right)
            +z^2\phi_{G\setminus 1\setminus 2} }
            \\
            &= -z^2 \frac{\phi_{G} - z\phi_{G\setminus 2}- z^{-1}\phi_{G\setminus 1}
            +\phi_{G\setminus 1\setminus 2} }{\phi_{G}
            - z\left(\phi_{G\setminus 1}+\phi_{G\setminus 2}\right)
            +z^2\phi_{G\setminus 1\setminus 2} }.
        \end{align*}
        The expression for $S_{22}(G^{+1},z)$ is obtained by swapping $1$ and $2$.
        Finally, we use \eqref{shift-terminals} to obtain the expressions for the elements of $S(G,z)$ from $S(G^{+1},z)$.
    \end{proof}

    \begin{remark}
        For a graph with both terminals located on the same vertex, the reflection and transmission coefficients are obtained by setting $\phi_{G\setminus 2}=\phi_{G\setminus 1}$ and $\phi_{G\setminus 1 \setminus 2}=0$:
        \begin{equation} \label{same-vertex-terminal}
            S_{11}(G,z) = -\frac{\phi_{G} - (z+z^{-1})\phi_{G\setminus 1}}{\phi_{G}
            - 2z\phi_{G\setminus 1}};
            \qquad
            S_{12}(G,z) = (z-z^{-1})\frac{\phi_{G\setminus 1}}{\phi_{G}
            - 2z\phi_{G\setminus 1}}.
        \end{equation}
    \end{remark}

    \subsection{Perfect transmission}

    For a graph to be applicable to quantum computing, it should be the case that it exhibits perfect transmission at the operating momentum $k$. 
    In the case of two-terminal graphs, this is simply the condition that $|S_{12}(G,e^{ik})|=1$, or equivalently, $S_{11}(G,e^{ik})=S_{22}(G,e^{ik})=0$.
    In what follows, we investigate explicitly when a graph exhibits perfect transmission at momentum $k$. 

    \begin{proposition} \label{prop:pt}
        Let $k \in (-\pi,0)$ and $\epsilon=2\cos(k)$.
        The following conditions are necessary for a graph $G$ with two terminals to have perfect transmission at momentum $k$:
        \begin{itemize}
            \item $\phi_{G\setminus1}(\epsilon)=\phi_{G\setminus2}(\epsilon)$; and
            \item $\phi_{G}(\epsilon)-\epsilon\phi_{G\setminus1}(\epsilon)+\phi_{G\setminus1\setminus 2}(\epsilon) = 0$.
        \end{itemize}
        If, further, at least one of $\phi_{G}(\epsilon), \phi_{G\setminus1}(\epsilon)$ or $\phi_{G\setminus1\setminus 2}(\epsilon)$ are nonzero then graph $G$ has perfect transmission at momentum $k$. 
    \end{proposition}

    \begin{proof}
        We will consider the numerator and denominator polynomials separately.
        Let 
        \begin{align*}
            n_{11}(z) &= \phi_{G}(z+z^{-1}) - z^{-1}\phi_{G\setminus 1}(z+z^{-1}) - z\phi_{G\setminus 2}(z+z^{-1})
            +\phi_{G\setminus 1\setminus 2}(z+z^{-1})
            \\
            d(z) &= \phi_{G}(z+z^{-1})
            - z\left(\phi_{G\setminus 1}(z+z^{-1})+\phi_{G\setminus 2}(z+z^{-1})\right)
            +z^2\phi_{G\setminus 1\setminus 2}(z+z^{-1}),
        \end{align*}
        so $S_{11}(G,z) = -n_{11}(z)/d(z)$. 
        Observe that the $n_{11}(z)$ and $d(z)$ are pole-free for $z$ on the unit circle, as the only possible pole is located at $z=0$.
        Hence, perfect transmission at momentum $k$ is achieved if and only if the order of the zero in the numerator exceeds that in the denominator at $z=e^{ik}$. 
        For there to be a zero in the numerator at all, we must have 
        \begin{equation} \label{pt-poly-nec}
            \phi_{G} - \cos(k) \left(\phi_{G\setminus 1} + \phi_{G\setminus 2}\right) + \phi_{G\setminus1\setminus 2}
            +
            i \sin(k) \left(\phi_{G\setminus 1} - \phi_{G\setminus 2}\right) =0,
        \end{equation}
        with each characteristic polynomial evaluated at $2\cos(k)$. 
        Since $H(G)$ and its submatrices are Hermitian, the characteristic polynomials are real for real inputs. 
        As such we obtain the two necessary conditions from the real and imaginary parts of the equation.
        
        Now, supposing the numerator is indeed zero at $z=e^{ik}$, to achieve perfect transmission, it is sufficient that the denominator is nonzero at this value.
        The denominator expression is 
        \[
            d(e^{ik})=\phi_{G} - \cos(k) \left(\phi_{G\setminus 1} + \phi_{G\setminus 2}\right) + \cos(2k)\phi_{G\setminus1\setminus 2}
            +
            i \left(\sin(2k)\phi_{G\setminus1\setminus 2}- \sin(k) \left(\phi_{G\setminus 1} + \phi_{G\setminus 2}\right) \right).
        \]
        If \eqref{pt-poly-nec} satisfied, then this may be written as 
        \[
            d(e^{ik})=(\cos(2k)-1)\phi_{G\setminus1\setminus 2}
            +
            2i\sin(k) \left(\cos(k)\phi_{G\setminus1\setminus 2}- \phi_{G\setminus 1} \right).
        \]
        Again, the characteristic polynomials are real, and $k \in (-\pi,0)$, so the denominator will be zero if and only if both $\phi_{G\setminus1\setminus 2}(2\cos(k))=0$ and $\phi_{G\setminus1}(2\cos(k))=0$. 
        With condition \eqref{pt-poly-nec}, this would imply that all four characteristic polynomials are zero at $2\cos(k)$. 
        Contrapositively, if one of these polynomials is nonzero then the denominator is nonzero, and so we have perfect transmission.
    \end{proof}

    \begin{corollary}
        If the graph $G$ with real weights has perfect transmission at momentum $k \in (-\pi,0)$ and one of $\phi_{G}, \phi_{G\setminus1}, \phi_{G\setminus2}$ or $\phi_{G\setminus1\setminus 2}$ are nonzero at $\epsilon = 2\cos(k)$ then the transmission phase is equal to
        \begin{equation} \label{S12-charpoly}
            S_{12}(G,e^{ik}) = \frac{e^{-ik}\phi_{G\setminus 1\setminus 2}(\epsilon)-\phi_{G\setminus 1}(\epsilon)}{\left|e^{-ik}\phi_{G\setminus 1\setminus 2}(\epsilon)-\phi_{G\setminus 1}(\epsilon)\right|}.
        \end{equation}
    \end{corollary}

    \begin{proof}
        From \eqref{S12}, we observe that $\sqrt{\phi_{G\setminus 1}\phi_{G\setminus 2}-\phi_{G}\phi_{G\setminus 1\setminus 2}}$ is a real number, so the phase is entirely determined by the denominator of $S_{12}(z)$, provided the denominator is nonzero. 
        As above, we write $d(z)$ for this denominator.
        When the necessary conditions for perfect transmission are satisfied,
        \begin{align*}
            d(e^{ik})&=\phi_{G}
                - \cos(k)\left(\phi_{G\setminus 1}+\phi_{G\setminus 2}\right)
                +\cos(2k)\phi_{G\setminus 1\setminus 2} 
                +i \left(\sin(2k)\phi_{G\setminus 1\setminus 2} 
                - \sin(k)\left(\phi_{G\setminus 1}+\phi_{G\setminus 2}\right) \right)
                \\
                &=\left(\cos(2k)-1\right)\phi_{G\setminus 1\setminus 2} 
                +i \left(\sin(2k)\phi_{G\setminus 1\setminus 2} 
                - \sin(k)\left(\phi_{G\setminus 1}+\phi_{G\setminus 2}\right) \right)
                \\
                &=2\sin(k)\left(-\sin(k)\phi_{G\setminus 1\setminus 2} 
                +i \left(\cos(k)\phi_{G\setminus 1\setminus 2} 
                - \phi_{G\setminus 1} \right)\right)
                \\
                &=2i\sin(k)
                \left(e^{ik}\phi_{G\setminus 1\setminus 2} 
                - \phi_{G\setminus 1} \right).
        \end{align*}
        We see that the factor of $2i\sin(k)$ cancels with that in the numerator, so
        \[
            S_{12}(G,e^{ik}) = \frac{\sqrt{\phi_{G\setminus 1}\phi_{G\setminus 2}-\phi_{G}\phi_{G\setminus 1\setminus 2}} 
            }{
                e^{ik}\phi_{G\setminus 1\setminus 2} 
                - \phi_{G\setminus 1}
            }.
        \]
        The result follows from the fact that $|S_{12}(G,e^{ik})|=1$. 
        
    \end{proof}

\section{Graphs in parallel} \label{sec:parallel}

In this section, we recast the $S$-matrix elements in terms of a different 
set of quantities, which resemble the admittance of an electrical circuit. 
These will allow us to derive the $S$-matrix elements for a two-terminal graph comprised of parallel two-terminal graphs in a convenient form. 
We then give a geometric characterisation of the requirements for perfect transmission in terms of these quantities.
\begin{center}
    \it From now on, assume that $G$ has real weights.
\end{center}

\subsection{\texorpdfstring{The quantities $\mu$ and $\nu$}{The quantities mu and nu}} \label{ssec:munu}

Let us define the rational functions
\begin{equation} \label{mu-def}
    \mu_1(G,y) := y- \frac{\phi_{G\setminus 1}(y)}{\phi_{G\setminus 1\setminus 2}(y)},
    \qquad 
    \mu_2(G,y) := y- \frac{\phi_{G\setminus 2}(y)}{\phi_{G\setminus 1\setminus 2}(y)},
\end{equation}
and
\begin{equation} \label{nu-def}
    \nu(G,y) := \frac{\sqrt{\phi_{G\setminus 1}(y)\phi_{G\setminus 2}(y)-\phi_{G}(y)\phi_{G\setminus 1 \setminus 2}(y)}}{\phi_{G\setminus 1 \setminus 2}(y)},
\end{equation}
for any $y\in \mathbb{R}$ for which $\phi_{G\setminus 1 \setminus 2}(y) \neq 0$.
Note that these three quantities are real numbers; moreover, if the graph weights are rational then they can each be written as a ratio of polynomials in $y$ with integer coefficients. 

We will want to think of these three quantities as characterising the properties of a given graph $G$, at a given momentum $k$. 
Hence, where there is no ambiguity, we will write $\mu_1$ for $\mu_1(G, 2\cos(k))$, and so on. 
In terms of these quantities, the $S$-matrix elements are given by
\begin{equation} \label{s11-munu}
    S_{11}(G,e^{ik}) = - \frac{(\mu_1-e^{-ik})(\mu_2-e^{ik})-\nu^2}{(\mu_1-e^{-ik})(\mu_2-e^{-ik})-\nu^2}; \qquad S_{22}(G,e^{ik}) = - \frac{(\mu_1-e^{ik})(\mu_2-e^{-ik})-\nu^2}{(\mu_1-e^{-ik})(\mu_2-e^{-ik})-\nu^2};
\end{equation}
\begin{equation} \label{s12-munu}
    S_{12}(G,e^{ik}) = \frac{2i\sin(k) \, \nu}{(\mu_1-e^{-ik})(\mu_2-e^{-ik})-\nu^2}.
\end{equation}

We can now recast the results of the previous section on perfect transmission in terms of the quantities $\mu_1$, $\mu_2$ and $\nu$. 
We assume for now that $\phi_{G\setminus 1 \setminus 2}(2\cos(k))\neq 0$.
Then the necessary and sufficient conditions for perfect transmission stated in Proposition~\ref{prop:pt} are equivalent to 
\begin{gather}
    \mu_1 = \mu_2 =: \mu \label{symm-eq}
    \\
    \left(\frac{\nu}{\sin(k)}\right)^2 -  \left(\frac{\mu-\cos(k)}{\sin(k)}\right)^2 = 1 \label{pt-hyp}.
\end{gather}

In other words, two-terminal graphs exhibiting perfect transmission draw out a hyperbola in the $\mu$-$\nu$ plane.
Further, we parametrise the hyperbola by
\begin{equation} \label{pt-param-munu}
    \mu = \cos(k)-\sin(k)\cot(\theta) = \frac{\sin(\theta-k)}{\sin(\theta)}
    \qquad 
    \nu = -\sin(k)\csc(\theta) = -\frac{\sin(k)}{\sin(\theta)},
\end{equation}
for $\theta \in (-\pi, 0) \cup (0, \pi)$.
We have reused the symbol $\theta$ as one can confirm that this is none other than the argument of $S_{12}(G,e^{ik})$; that is,
\begin{equation} \label{pt-phase}
    S_{12}(G,e^{ik}) = \frac{\nu}{\mu-e^{-ik}} = e^{i\theta}.
\end{equation}
Figure \ref{fig:hyperbola-phase} gives a geometric picture of the relationship between $\theta$, $\mu$ and $\nu$.

\begin{figure}[ht]
    
    \centering

    \includegraphics{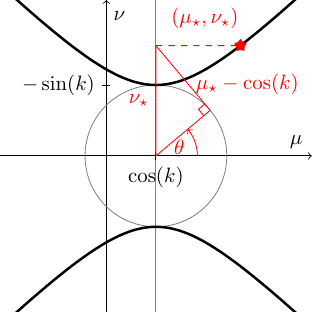}

    \caption{
        The hyperbola \eqref{pt-hyp} in the $\mu$-$\nu$ plane. The points of the hyperbola are parametrised by the phase of transmission $\theta \in (-\pi,0) \cup (0,\pi)$.
        }
    \label{fig:hyperbola-phase}

\end{figure}

\begin{remark} \label{rem:perfect-reflection}
    For certain applications, such as the R/T gadget introduced in \cite{childsMomentumSwitches2015}, it is 
    important to also consider perfect reflection. 
    We observe that the condition for perfect reflection is simply $\nu=0$ --  that is, the equivalent of the perfect transmission hyperbola for perfect reflection is the $\mu$-axis. In this case, the reflection coefficients are 
    \begin{equation}
        S_{11}(e^{ik}) = -\frac{\mu_2 - e^{ik}}{\mu_2 - e^{-ik}}
        \quad \text{and} \quad
        S_{22}(e^{ik}) = -\frac{\mu_1 - e^{ik}}{\mu_1 - e^{-ik}}.
    \end{equation}
\end{remark}

\subsection{Graphs in parallel}

Consider a sequence of graphs $(G_j)_{j=1}^n$, each with two terminal vertices labelled $1$ and $2$ respectively, and no other shared nodes; in other words, $G_{j_1} \cap G_{j_2} = \{1,2\}$. 
We will refer to their union as the parallel graph,
\[
    G_{\parallel} := \bigcup_{j=1}^n G_j.
\]
The edge weight of the edge connecting $1$ and $2$ in the parallel graph is defined to be the sum $\omega_{G_{\parallel}}(1,2) = \omega_{G_1}(1,2) + \dots + \omega_{G_n}(1,2)$. 
We will demonstrate a number of properties of the parallel graph. 
First, note that removing vertices $1$ and $2$ disconnects the graphs, so
\[
    \phi_{G_\parallel \setminus 1 \setminus 2}(y) = \prod_{j=1}^n \phi_{G_j \setminus 1 \setminus 2}(y).
\]

\begin{proposition}
    Let $(G_j)_{j=1}^n$ be a sequence of graphs with real weights such that $G_{j_1} \cap G_{j_2} = \{1,2\}$ for any $1\leq j_1,j_2 \leq n$. 
    Then 
    \[
        \mu_1(G_{\parallel},y) = \sum_{j=1}^n \mu_1(G_j,y),
        \quad 
        \mu_2(G_{\parallel},y) = \sum_{j=1}^n \mu_2(G_j,y)
        \quad 
        \text{and}
        \quad 
        \nu(G_{\parallel},y) = \sum_{j=1}^n \nu(G_j,y).
    \]
\end{proposition}

\begin{proof}
    It will be sufficient to prove that the functions are additive for a pair of graphs in parallel.
    For a path $P$, let $\omega(P)$ denote the product of edge weights of edges in that path. 

    Let us begin with the proof for $\mu_1$; in the case of $\mu_1$, the statement is also true for complex weights, as we show in what follows.
    The characteristic polynomial for a graph can be expressed in terms of those of its subgraphs via the Schwenk formula \cite{schwenkComputingCharacteristicPolynomial1974}, which is
    \[
        \phi_{G}(y) = y \phi_{G \setminus u}(y)
        - \sum_{\substack{v \in \mathcal{V}(G) \\ v \sim u}} |\omega(u,v)|^2 \phi_{G \setminus u \setminus v}(y)
        - 2\sum_{C \in \mathcal{C}(G) | u \in C} \Re\left(\omega(C)\right) \phi_{G \setminus C}(y),
    \]
    where $\mathcal{C}(G)$ is the set of cycles in $G$.
    For completeness, we include a proof of this statement in Appendix~\ref{schwenk}.
    Recasting this in terms of $\mu_1(G,y)$ yields
    \[
        \mu_1(G,y) = \sum_{\substack{v \in \mathcal{V}(G\setminus 1) \\ v \sim 2}} |\omega(2,v)|^2 
        \frac{\phi_{G \setminus 1 \setminus 2 \setminus v}(y)}{\phi_{G \setminus 1 \setminus 2}(y)}
        +2\sum_{C \in \mathcal{C}(G \setminus 1) | 2 \in C} \Re\left(\omega(C)\right) \frac{\phi_{G \setminus 1 \setminus C}(y)}{\phi_{G \setminus 1 \setminus 2}(y)}.
    \]
    For $G_{\parallel}$, each of these sums decomposes into equivalent sums over the subgraphs $G_1$ and $G_2$. 
    Upon doing so, we observe cancellations such as, for $v \in G_2$, 
    \[
        \frac{\phi_{G \setminus 1 \setminus 2 \setminus v}(y)}{\phi_{G \setminus 1 \setminus 2}(y)}
        =
        \frac{\phi_{G_1 \setminus 1 \setminus 2}(y)\phi_{G_2 \setminus 1 \setminus 2 \setminus v}(y)}{\phi_{G_1 \setminus 1 \setminus 2}(y)\phi_{G_2 \setminus 1 \setminus 2}(y)}
        =
        \frac{\phi_{G_2 \setminus 1 \setminus 2 \setminus v}(y)}{\phi_{G_2 \setminus 1 \setminus 2}(y)},
    \]
    which results in the following:
    \[
        \sum_{\substack{v \in \mathcal{V}(G\setminus 1) \\ v \sim 2}} |\omega(2,v)|^2 
        \frac{\phi_{G \setminus 1 \setminus 2 \setminus v}(y)}{\phi_{G \setminus 1 \setminus 2}(y)}
        =
        \sum_{\substack{v \in \mathcal{V}(G_1\setminus 1) \\ v \sim 2}} |\omega(2,v)|^2 
        \frac{\phi_{G_1 \setminus 1 \setminus 2 \setminus v}(y)}{\phi_{G_1 \setminus 1 \setminus 2}(y)}
        +
        \sum_{\substack{v \in \mathcal{V}(G_2\setminus 1) \\ v \sim 2}} |\omega(2,v)|^2 
        \frac{\phi_{G_2 \setminus 1 \setminus 2 \setminus v}(y)}{\phi_{G_2 \setminus 1 \setminus 2}(y)}.
    \]
    The same argument holds for the sum over cycles, as each cycle in $G_{\parallel}\setminus 1$ must be fully contained in one of $G_1$ or $G_2$.
    As a result, we obtain the result for $\mu_1(G_{\parallel},y)$. 
    The result for $\mu_2(G_{\parallel},y)$ follows by symmetry. 

    For $\nu(G,y)$, we make use of the sum-over-paths formula \cite[Chapter 4, Corollary 2.2]{godsilAlgebraicCombinatorics1993} \cite{schwenkRemovalcospectralSetsVertices1979},
    \begin{equation}\label{sumofpaths}
        \sqrt{\phi_{G \setminus u}(y) \phi_{G\setminus v}(y)-\phi_{G}(y) \phi_{G\setminus u\setminus v}(y)} = \sum_{P \in \mathcal{P}(G)_{uv}} \omega(P) \phi_{G\setminus P}(y).
    \end{equation}
    Here $\mathcal{P}(G_{\parallel})_{uv}$ is the set of paths from 1 to 2 in $G_{\parallel}$.
    For $G_{\parallel}$, the sum decomposes as follows:
    \[
        \sum_{P \in \mathcal{P}(G_{\parallel})_{12}} \omega(P) \phi_{G_{\parallel}\setminus P}(y)
        =
        \sum_{P \in \mathcal{P}(G_1)_{12}'} \omega(P) \phi_{G_{\parallel}\setminus P}(y)
        +
        \sum_{P \in \mathcal{P}(G_2)_{12}'} \omega(P) \phi_{G_{\parallel}\setminus P}(y)
        +
        \omega_{G_{\parallel}}(1,2) \phi_{G \setminus 1 \setminus 2}(y).
    \]
    We have used notation $\mathcal{P}(G)_{12}'$ to exclude the edge connecting $1$ and $2$.
    Since $\omega_{G_{\parallel}}(1,2) = \omega_{G_1}(1,2) + \omega_{G_2}(1,2)$, we restore the full sum
    \[
        \sum_{P \in \mathcal{P}(G_{\parallel})_{12}} \omega(P) \phi_{G_{\parallel}\setminus P}(y)
        =
        \sum_{P \in \mathcal{P}(G_1)_{12}} \omega(P) \phi_{G_{\parallel}\setminus P}(y)
        +
        \sum_{P \in \mathcal{P}(G_2)_{12}} \omega(P) \phi_{G_{\parallel}\setminus P}(y)
    \]
    From the above, we divide through by $\phi_{G \setminus 1 \setminus 2}(y)$ to obtain
    \[
        \nu(G_{\parallel},y) 
        =
        \sum_{P \in \mathcal{P}(G_1)_{12}} \omega(P) \frac{\phi_{G_1\setminus P}(y)}{\phi_{G_1 \setminus 1 \setminus 2}(y)}
        +
        \sum_{P \in \mathcal{P}(G_2)_{12}} \omega(P) \frac{\phi_{G_2\setminus P}(y)}{\phi_{G_2 \setminus 1 \setminus 2}(y)}
    \]
    These sums are just $\nu(G_1,y)$ and $\nu(G_2,y)$.
\end{proof}

\begin{remark}
    The quantities $\mu_1$, $\mu_2$ and $\nu$ are related to the admittance, the inverse of the impedance in an electrical circuit.
    Indeed, from \cite{pozarMicrowaveEngineering2012}, the $S$-matrix can be written in terms of the admittance matrix $Y$ as $S=(I+Y)^{-1}(I-Y)$.
    Expanding this gives 
    \[
        S(G, e^{ik}) = 
        -\left(
            \begin{array}{cc}
                \dfrac{(Y_{11}-1)(Y_{22}+1)-Y_{12}^2}{(Y_{11}+1)(Y_{22}+1)-Y_{12}^2} 
            & 
                \dfrac{2 Y_{12}}{(Y_{11}+1)(Y_{22}+1)-Y_{12}^2} 
            \\
                \dfrac{2 Y_{12}}{(Y_{11}+1)(Y_{22}+1)-Y_{12}^2} 
            & 
                \dfrac{(Y_{11}+1)(Y_{22}-1)-Y_{12}^2}{(Y_{11}+1)(Y_{22}+1)-Y_{12}^2}
            \end{array}
        \right).
    \]
    The transformation $Y_{12} = Y_{21} = -\frac{\nu}{i \sin{k}}$, $Y_{11} = \frac{\mu_2-\cos(k)}{i \sin{k}}$ and $Y_{22} = \frac{\mu_1-\cos(k)}{i \sin{k}}$ gives the direct change of variables.
\end{remark}

Consider the case where $G$ is a tree graph; that is, $G$ has no cycles. 
In this specific case, there is an $O(n\log^2(n))$ algorithm for determining the characteristic polynomial's coefficients, see \cite{furerEfficientComputationCharacteristic2014}, allowing for a fast calculation of the $S$-matrix. 
However, a quantum walk with two terminals and whose finite part is a tree graph reduces to a collection of tree graphs $T_1, T_2, \dots, T_N$ connected in series with quantum wires, each with both of its terminals on the same vertex. 
In other words, the entire system $\widehat{G}$ consists of a single doubly-infinite line, with `branches' at $N$ selected vertices. 
This significantly restricts the $S$-matrix; in particular, if we have perfect transmission then the phase of transmission is trivial, as can be deduced from \eqref{same-vertex-terminal}. 
Combining tree graphs in parallel, however, allows us to construct graphs with cycles, while still making use of the fast computation method available to tree graphs.

\subsection{Graphs in series}

For completeness, we give here the parameters $\mu_1$, $\mu_2$ and $\nu$ for graphs combined in series. 
Let ${G_1}$, ${G_2}$ denote graphs with two terminals, and let ${G_1}{G_2}$ denote the two-terminal graph obtained by identifying terminal 2 of ${G_1}$ with terminal 1 of ${G_2}$. 
The corresponding $S$-matrix for this setup can be obtained from standard scattering theory, as in \cite{childsUniversalComputationMultiparticle2013}. 
Alternatively, the formula for a characteristic polynomial of a graph with a cutpoint (see e.g. \cite[Chapter 4, Corollary 3.3]{godsilAlgebraicCombinatorics1993}) may be used to obtain the result directly from the characteristic polynomials of the constituent graphs, which we show below.

\begin{proposition}
    Consider a graph ${G_1}{G_2}$ as described above. 
    We have 
    \begin{equation} \label{mu-series}
        \mu_1({G_1}{G_2},y) = \mu_1({G_2},y) + 
        \frac{\nu({G_2},y)^2}{y-\mu_1({G_1},y)-\mu_2({G_2},y)};
    \end{equation}
    \begin{equation}
        \mu_2({G_1}{G_2},y) = \mu_2({G_1},y) + 
        \frac{\nu({G_1},y)^2}{y-\mu_1({G_1},y)-\mu_2({G_2},y)};
    \end{equation}
    \begin{equation}\label{nu-series}
        \nu({G_1}{G_2},y) = \frac{\nu({G_1},y)\nu({G_2},y)}{y-\mu_1({G_1},y)-\mu_2({G_2},y)} .
    \end{equation}
\end{proposition}

\begin{proof}
    From \cite[Chapter 4, Corollary 3.3]{godsilAlgebraicCombinatorics1993}
    \[
        \phi_{{G_1}{G_2}} = \phi_{{G_1} \setminus 2} \phi_{{G_2}}  + \phi_{{G_1}} \phi_{{G_2} \setminus 1}  - y \, \phi_{{G_1} \setminus 2} \phi_{{G_2} \setminus 1}.
    \]
    As a result, we have
    \begin{align*}
        \mu_1({G_1}{G_2},y) &= y - 
        \frac{\phi_{{G_1} \setminus 1 \setminus 2} \phi_{{G_2}}  + \phi_{{G_1} \setminus 1} \phi_{{G_2} \setminus 1}  - y \, \phi_{{G_1} \setminus 1 \setminus 2} \phi_{{G_2} \setminus 1}}{\phi_{{G_1} \setminus 1 \setminus 2} \phi_{{G_2} \setminus 2}  + \phi_{{G_1} \setminus 1} \phi_{{G_2} \setminus 1 \setminus 2}  - y \, \phi_{{G_1} \setminus 1 \setminus 2} \phi_{{G_2} \setminus 1 \setminus 2}}
        \\
        &= y - 
        \frac{\phi_{{G_2}}/\phi_{{G_2} \setminus 1 \setminus 2}   + \left(\phi_{{G_1} \setminus 1}/\phi_{{G_1} \setminus 1 \setminus 2}\right) \left( \phi_{{G_2} \setminus 1}/\phi_{{G_2} \setminus 1 \setminus 2}\right)   - y \, \phi_{{G_2} \setminus 1}/\phi_{{G_2} \setminus 1 \setminus 2} }{\phi_{{G_2} \setminus 2}/\phi_{{G_2} \setminus 1 \setminus 2}   + \phi_{{G_1} \setminus 1}/\phi_{{G_1} \setminus 1 \setminus 2}   - y}.
    \end{align*}
    This can then be transformed using the identities $\phi_{{G_1}\setminus 1}/\phi_{{G_1}\setminus 1 \setminus 2} = y-\mu_1({G_1},y)$ and $\phi_{G_1}/\phi_{{G_1}\setminus 1 \setminus 2} = \left(y-\mu_1({G_1},y)\right)\left(y-\mu_2({G_1},y)\right) - \nu({G_1},y)^2$.
    The resulting expression is 
    \begin{align*}
        &y - 
        \frac{\left(y-\mu_1({G_2},y)\right)\left(y-\mu_2({G_2},y)\right) - \nu({G_2},y)^2   + (y-\mu_1({G_1},y)) (y-\mu_1({G_2},y))   - y (y-\mu_1({G_2},y)) }{y-\mu_1({G_1},y)-\mu_2({G_2},y)}
        \\
        &\qquad = y - (y-\mu_1({G_2},y)) + 
        \frac{\nu({G_2},y)^2}{y-\mu_1({G_1},y)-\mu_2({G_2},y)} 
        \\
        &\qquad = \mu_1({G_2},y) + 
        \frac{\nu({G_2},y)^2}{y-\mu_1({G_1},y)-\mu_2({G_2},y)},
    \end{align*}
    yielding \eqref{mu-series}.

    For $\nu({G_1}{G_2},y)$, we invoke the sum of paths formula \eqref{sumofpaths}; since each path between the two terminals of ${G_1}{G_2}$ must pass through the shared vertex, each such path can be decomposed into a path fully contained within ${G_1}$, and one fully contained within ${G_2}$. As a result,
    \begin{align*}
        &\sum_{P \in \mathcal{P}({G_1}{G_2})_{12}} \omega_{{G_1}{G_2}}(P) \phi_{{G_1}{G_2}\setminus P}(y) 
        \\
        &\hspace{2cm}= \sum_{P_{G_1} \in \mathcal{P}({G_1})_{12}}
         \sum_{P_{G_2} \in \mathcal{P}({G_2})_{12}} \omega_{{G_1}}(P_{G_1}) \, \omega_{{G_2}}(P_{G_2}) \phi_{{G_1}{G_2}\setminus P_{G_1}\setminus P_{G_2}}(y) 
        \\
        &\hspace{2cm}= \sum_{P_{G_1} \in \mathcal{P}({G_1})_{12}}\sum_{P_{G_2} \in \mathcal{P}({G_2})_{12}} \omega_{{G_1}}(P_{G_1}) \, \omega_{{G_2}}(P_{G_2}) \phi_{{G_1}\setminus P_{G_1}}(y) \phi_{{G_2}\setminus P_{G_2}}(y) 
        \\
        &\hspace{2cm}= \left(\sum_{P_{G_1} \in \mathcal{P}({G_1})_{12}} \omega_{{G_1}}(P_{G_1})  \phi_{{G_1}\setminus P_{G_1}}(y) \right) \left(\sum_{P_{G_2} \in \mathcal{P}({G_2})_{12}}  \omega_{{G_2}}(P_{G_2})\phi_{{G_2}\setminus P_{G_2}}(y) \right).
    \end{align*}
    This constitutes the numerator of $\nu({G_1}{G_2},y)$; the denominator is identical to that of $\mu_1({G_1}{G_2},y)$, and so we obtain \eqref{nu-series}.
\end{proof}

\section{Further properties}

\subsection{Poles}

In this section we investigate the case where $\phi_{G\setminus 1 \setminus 2}(y_0) = 0$ -- that is, $y_0$ is an eigenvalue of $H(G\setminus 1 \setminus 2)$. 
Let $m_{G}(\lambda)$ denote the multiplicity of eigenvalue $\lambda$ in $H(G)$.
The matrix $H(G)$ is real symmetric; as a result, the eigenvalues satisfy the interlacing theorem \cite[Theorem~5.3]{godsilAlgebraicCombinatorics1993}. 
Eigenvalue interlacing implies that $-1 \leq \left(m_{G\setminus 1}(y_0) - m_{G\setminus 1 \setminus 2}(y_0)\right) \leq 1$, and so
\[
    \mu_1(G,y) \sim (y-y_0)^{m} \text{ as } y\to y_0, \qquad  m \in \{-1,0,1\},
\]
similarly for $\mu_2(G,y)$.
Using eigenvalue interlacing again, we have $m_{G\setminus 1 \setminus 2}(y_0)-2 \leq m_{G}(y_0) \leq m_{G\setminus 1 \setminus 2}(y_0)+2$; this implies that
\[
    \nu(G,y) \sim (y-y_0)^{m'} \text{ as } y\to y_0, \qquad  m' \in \{-1,0,1\}.
\]
Hence, there is at worst a simple pole at $y_0$ in each of $\mu_1$, $\mu_2$ and $\nu$. 
Excepting the case of a pole, the analysis of the previous section is unchanged; one must simply replace $\mu_1$, $\mu_2$ and $\nu$ by their limiting values. 

If $\mu_1(G,y)$, $\mu_2(G,y)$ and/or $\nu(G,y)$ contain a simple pole, then 
we can repeat the previous analysis to obtain necessary and sufficient conditions for perfect transmission. 
Write the Laurent expansion around $y=y_0$ of these quantities as follows:
\[
    \nu(G,y) = \sum_{n=-1}^\infty (y-y_0)^n\nu_{n};
\]
similarly for $\mu_1(G,y)$ and $\mu_2(G,y)$. 
Consider the $S$-matrix at momentum $k$, and set $y=2\cos(k)$. 
Looking at the leading order, we will have $S_{11}(G,e^{ik}) = -1$ (i.e. perfect reflection) unless 
\begin{equation} \label{pole-pt-1}
    \mu_{1,-1}\mu_{2,-1} = \nu_{-1}^2.
\end{equation}
Next, looking at the imaginary part of the numerator at order $-1$, we obtain a necessary condition for perfect transmission:
\begin{equation} \label{pole-pt-2}
    \mu_{1,-1} = \mu_{2,-1}.
\end{equation}
From above, this implies that $\mu_{1,-1} = \pm \nu_{-1}$.
Hence, if there is a pole in one of the three functions, and we are to have perfect transmission, then all three must have a pole at the same value, and the residues must be equal up to a sign. 

From the real part of the numerator, we obtain 
\[
    \mu_{1,-1}\mu_{2,0}+\mu_{1,0}\mu_{2,-1}-2\nu_{-1}\nu_{0} = 2\cos(k) \, \mu_{1,-1},
\]
or,
\begin{equation} \label{pole-pt-3}
    (\mu_{2,0}+\mu_{1,0})\mp 2\nu_{0} = 2\cos(k).
\end{equation}
Now, turning to the denominator, we have 
\[
    (\mu_1-e^{-ik})(\mu_2-e^{-ik}) - \nu^2 =  
    \mu_{1,-1}\mu_{2,0}+\mu_{1,0}\mu_{2,-1}-2\nu_{-1}\nu_{0} - 2z^{-1}\mu_{1,-1} + O(1)
\]
Substituting in the above necessary condition, this becomes
\[
    2i\sin(k)\mu_{1,-1} + O(1).
\]
As established above, $\mu_{1,-1}$ is nonzero, so the equations \eqref{pole-pt-1}, \eqref{pole-pt-2}, and \eqref{pole-pt-3} constitute necessary and sufficient conditions for perfect transmission in the presence of a pole. 
Further the phase of transmission is given by 
\[
    S_{12}(G,z) = \frac{2i\sin(k)\nu_{-1}}{2i\sin(k)\mu_{1,-1}} = \pm1.
\]
Hence, in the case of a pole, the phase of transmission is trivial.

\subsection{The effective length}

So far, we have investigated how the phase of transmission $e^{i\theta}$ is related to the quantities $\nu$, $\mu_1$ and $\mu_2$.
In the context of quantum computation, we must also calculate the effective length. 
In essence, the effective length tells us the time taken for a wavepacket of momentum $k$ travelling through the graph; this is essential for calibrating the unitary gate, as the wavepackets associated to both the $\ket{0}$ and $\ket{1}$ states must move in unison. 
In the following, we derive a formula for the effective length of a graph that exhibits perfect transmission in terms of $\mu_1$, $\mu_2$ and $\nu$, and their derivatives. 

The effective length between the two terminals is defined as
\newcommand{\dd}{\mathrm{d}}
\newcommand{\ddk}[1]{\frac{\dd #1}{\dd k}}
\[
    \ell_{12}(G,k):=\ddk{} \arg S_{12}(G,e^{ik}) = -i \ddk{} \ln \frac{S_{12}(G,e^{ik})}{|S_{12}(G,e^{ik})|};
\]
here the branch of the logarithm is irrelevant due to the derivative. 
Now assume that $G$ exhibits perfect transmission at $k=k_\star$. 
Since 1 is the global maximum value of $|S_{12}(G,e^{ik})|$ as a function of $k$, the derivative, and consequently the derivative of $\ln |S_{12}(G,e^{ik})|$ vanishes at $k=k_\star$. 

For what remains, recall from \eqref{s12-munu}
\[
    S_{12}(G,e^{ik}) = \frac{2i\sin(k)\nu}{(\mu_1-e^{-ik})(\mu_2-e^{-ik})-\nu^2},
\]
where $\nu=\nu(G,2\cos(k))$, $\mu_1=\mu_1(G,2\cos(k))$ and $\mu_2=\mu_2(G,2\cos(k))$.
We begin by simply taking the logarithm, and then the derivative, yielding
\begin{align}
    \ddk{} \ln S_{12}(G,e^{ik})
    &= \ddk{}\ln (2i\sin(k)\nu) - \ddk{}\ln((\mu_1-e^{-ik})(\mu_2-e^{-ik})-\nu^2) \nonumber
    \\
    &= \cot(k) +  \frac{\dot\nu}{\nu} - \frac{(\dot\mu_1+i e^{-ik})(\mu_2-e^{-ik})+(\mu_1-e^{-ik})(\dot\mu_2+ie^{-ik})-2\nu\dot\nu}{(\mu_1-e^{-ik})(\mu_2-e^{-ik})-\nu^2}, \label{eff-length-initial}
\end{align}
where $\dot{\vphantom{o}}$ here denotes the derivative with respect to $k$. 
Let us evaluate this at $k=k_\star$, in which case we have the simplification $\mu_1 = \mu_2$, and 
\[
    (\mu_1-e^{-ik_\star})(\mu_2-e^{ik_\star})-\nu^2 = 0.
\]
As a result, the denominator in \eqref{eff-length-initial} becomes
\[
    (\mu_1-e^{-ik_\star})(\mu_2-e^{-ik_\star})-\nu^2 = 2i\sin(k_\star)(\mu_1-e^{-ik_\star}).
\]
Hence, there are a number of cancellations
\begin{align*}
    \ddk{} \ln S_{12}(G,e^{ik}) \bigg|_{k=k_\star}
    &= \cot(k_\star) +  \frac{\dot\nu}{\nu} 
    -
    \frac{\dot\mu_1+ie^{-ik_\star}}{2i\sin(k_\star)}
    -
    \frac{\dot\mu_2+ie^{-ik_\star}}{2i\sin(k_\star)}
    +
    \frac{2\nu\dot\nu}{2i\sin(k_\star)(\mu_1-e^{-ik_\star})}
\end{align*}
We now apply the chain rule to write the $k$ derivatives in terms of $y$-derivatives. 
Let $\vphantom{a}'$ denote the derivative with respect to $y$.
Then $\dot\nu = -2\sin(k)\nu'$; similarly for $\dot\mu_1$ and $\dot\mu_2$. As a result, 
\begin{align*}
    \ddk{} \ln S_{12}(G,e^{ik}) \bigg|_{k=k_\star}
    &= i -2\sin(k_\star)  \frac{\nu'}{\nu} 
    +
    \frac{\mu_1'+\mu_2'}{i}
    -
    \frac{2\nu\nu'}{i(\mu_1-e^{-ik_\star})}
\end{align*}
Finally, we recall the formulae \eqref{pt-phase} and \eqref{pt-param-munu}. Multiplying the result by $-i$, we obtain the effective length
\begin{equation} \label{efflength}
    \ell_{12}(G,k_\star)
    =  -(\mu_1' + \mu_2') + 2\cos(\theta)\nu'  + 1.
\end{equation}

We emphasize that the above formula is only valid when we have perfect transmission. However, $\mu_1$, $\mu_2$ and $\nu$ each combine additively for parallel graphs, and by extension so do their derivatives. 
Thus, the calculation of the effective length of a combined graph can be simplified, regardless of whether the constituent graphs themselves exhibit perfect transmission. 
Indeed, suppose we combine graphs $G_1,G_2,\dots,G_N$ in parallel to form $G_{\parallel}$, which exhibits perfect transmission at $k_\star$. 
Then, the effective length between the two terminals of $G_{\parallel}$ is given by 
\[
    \ell_{12}(G_{\parallel},k_\star) = -\sum_{j=1}^N (\mu_{1}'(G_j,2\cos(k_\star)) + \mu_{2}'(G_j,2\cos(k_\star))) + 2\cos(\theta) \sum_{j=1}^N \nu'(G_j,2\cos(k_\star)) + 1.
\]

\section{Examples}

In this section, we give examples of structures that can be formed by composing graphs in parallel, and the corresponding calculations using the quantities $\mu_1$, $\mu_2$ and $\nu$, at various values of the operating momentum $k$. 

\begin{example}
    Consider the case $k=-\pi/4$, which is the proposed operating momentum in \cite{childsUniversalComputationQuantum2009}.
    For a path graph of length $\ell$, perfect transmission is automatic, and the phase of transmission is simply $e^{ik\ell}$. 
    This allows us to read the values of $\mu$ and $\nu$ from the parametrisation of the perfect transmission hyperbola \eqref{pt-param-munu}, yielding the results below.
    \[
        \begin{array}{c|cc}
            \ell & \mu & \nu \\
            \hline
            1 & 0 & -1 \\
            2 & \sqrt{2}/2 & -\sqrt{2}/2 \\
            3 & \sqrt{2} & -1 \\
            4 & \infty & -\infty \\
        \end{array}
        \qquad 
        \begin{array}{c|cc}
            \ell & \mu & \nu \\
            \hline
            5 & 0 & 1 \\
            6 & \sqrt{2}/2 & \sqrt{2}/2 \\
            7 & \sqrt{2} & 1 \\
            8 & \infty & \infty \\
        \end{array}
    \]
    Combining paths in parallel corresponds to taking non-negative integer linear combinations of the above values.
    In this way, we can build a graph of parallel paths which has perfect transmission at momentum $-\pi/4$.
    An example of such a graph is given by four paths of length 3 and nine paths of length 5; see Figure~\ref{fig:parallelpaths}. 
    In this case, $\mu = 4\sqrt{2}$ and $\nu = 5$, which can be confirmed to satisfy \eqref{pt-hyp}.

    From Figure~\ref{fig:parallelpaths}, we also see visually that two paths of length 5, and one path of length 3 combined in parallel gives a graph with perfect transmission. 
    Additionally, from \eqref{rem:perfect-reflection}, we see that the setup made by composing one of each length path, thus creating a cycle of length 8, is perfectly reflective at $k=-\pi/4$.  
    
    \begin{figure}[ht]
        \centering
        \includegraphics{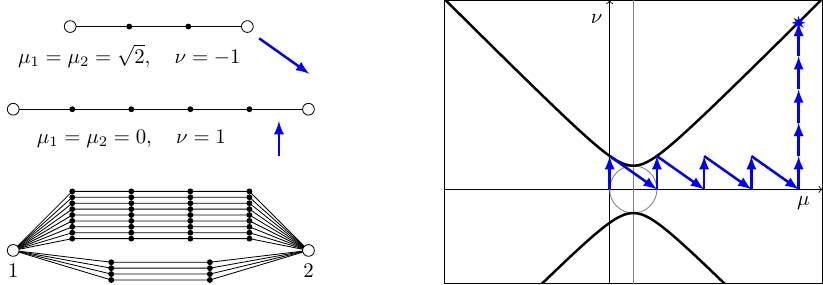}
        \caption{A visual explanation of the perfect transmission property of a two-terminal graph consisting of parallel paths at momentum $k=-\pi/4$. Each path is represented as a vector in the $\mu$-$\nu$ plane, and their parallel composition corresponds to the sum of these vectors, landing precisely on the perfect transmission hyperbola. }
        \label{fig:parallelpaths}
    \end{figure}

\end{example}

\begin{example}
    Consider the graphs below. 
    \begin{center}
        \includegraphics{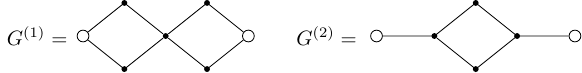}
    \end{center}
    For these graphs, the values of $\mu$ and $\nu$ are given explicitly by
    \begin{gather}
        \mu(G^{(1)},y) = \frac{2(y^2-2)}{y(y^2-4)} ;
        \qquad 
        \nu(G^{(1)},y) = \frac{4}{y(y^2-4)};
        \\
        \mu(G^{(2)},y) = \frac{2y}{y^2-1} ;
        \qquad 
        \nu(G^{(2)},y) = \frac{2y}{y^2-1}.
    \end{gather}

    Now define $G$ as the graph obtained by adjoining these graphs in parallel.
    The functions $\mu(G,y)$ and $\nu(G,y)$ of this graph are simply the sum of those for $G^{(1)}$ and $G^{(2)}$ given above.
    Evaluating \eqref{pt-hyp} reveals that this graph has perfect transmission only if $y=\pm\sqrt{3}$, that is, $k=-\pi/6$ or $k=-5\pi/6$.
    See Figure~\ref{fig:ex-pi6} for a graphical explanation. 
    At this point, $(\mu(G,\sqrt{3}),\nu(G,\sqrt{3})) = (1/\sqrt{3},-1/\sqrt{3})$, resulting in a phase angle of $\theta = 2\pi/3$. 
    Finally, using formula \eqref{efflength}, the effective length at $k=-\pi/6$ can be found to be 11. 
    \begin{figure}[ht]
        \centering
        \raisebox{0.5\height}{\includegraphics[scale=1.5]{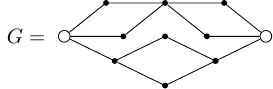}}
        \hfill
        {\setlength{\fboxsep}{0pt}%
        \fbox{\includegraphics[scale=1.0]{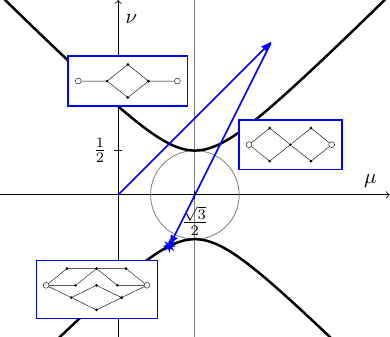}}}
        \caption{Left: the graph obtained by combining graphs $G^{(1)}$ and $G^{(2)}$ in parallel. Right: a graphical explanation of perfect transmission exhibited by $G$ at $k=-\pi/6$. The vector $(\mu,\nu)$ for each graph is plotted as a blue line. The sum of the two vectors for $G^{(1)}$ and $G^{(2)}$ lands on the perfect transmission hyperbola, indicated by the blue star.}
        \label{fig:ex-pi6}
    \end{figure}
\end{example}

\begin{example}

The setup depicted in Figure~\ref{fig:pi4phasegate} produces a phase gate with irrational phase for $k=-\pi/4$. 
Denoting the graph connecting $1_{\mathrm{in}}$ and $1_{\mathrm{out}}$ by $G$, we find
\[
    \mu(G,2\cos\left(\tfrac{\pi}{4}\right)) = \frac{1}{4} + \frac{\sqrt{2}}{4},
    \qquad
    \nu(G,2\cos\left(\tfrac{\pi}{4}\right)) = -\frac34.
\]
From \eqref{pt-hyp}, the graph $G$ exhibits perfect transmission at $k=-\pi/4$, and the resulting 
value of the transmission phase is, from \eqref{pt-phase},
\[
    S_{12}(G,e^{-i\pi/4}) = -\frac13-\frac{2\sqrt{2}}{3}i.
\]
Further, the derivatives of $\mu$ and $\nu$ are 
\[
    \mu'(G,2\cos\left(\tfrac{\pi}{4}\right)) = -\frac{5}{2} + \frac{\sqrt{2}}{4}; 
    \qquad \nu'(G,2\cos\left(\tfrac{\pi}{4}\right)) = -\frac{3\sqrt{2}}{4}.
\]
Using \eqref{efflength}, we obtain an effective length of 6. 
Therefore, if a quantum state, encoded using dual rail encoding as a pair of wavepackets, is sent through the setup in Figure~\ref{fig:pi4phasegate}, the resulting phase difference of the two wavepackets will be approximately $\psi = \theta - (-3\pi/4) \mod 2\pi$.

\begin{figure}[ht]
    \centering
    \includegraphics{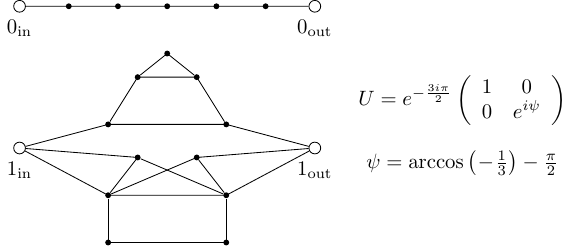}
    \caption{Phase gate for operating momentum $k=-\pi/4$.
    The resulting phase angle is equal to $\psi = \arccos(-1/3) - \pi/2\approx 0.108173\pi$.}
    \label{fig:pi4phasegate}
\end{figure}
    
\end{example}

\section{Conclusion and Discussion}

In this paper, we developed a characteristic-polynomial description of
two-terminal scattering for continuous-time quantum walks on graphs.  The
two-terminal $S$-matrix was expressed in terms of the characteristic
polynomial of the finite graph and those of its vertex-deleted subgraphs.
We also introduced the quantities $\mu_1$, $\mu_2$, and $\nu$, which
reorganize these formulae into an admittance-like form.  In this
parametrization, perfect transmission at a fixed momentum is described by
$\mu_1=\mu_2$ and, in the symmetric case, by a hyperbola in the
$(\mu,\nu)$-plane.  The point on this hyperbola determines the transmission
phase.

The main benefit of this formulation is that parallel composition becomes
vector addition in the $(\mu_1,\mu_2,\nu)$-space.  Hence, the search for
graphs with prescribed transmission properties is reduced to a geometric
vector-sum problem for smaller building blocks.

There are several natural directions for further work.  The first is to
extend the present formalism to graphs with more than two terminals.  This is
necessary for a direct treatment of general quantum gates.  The analogy with
electrical admittance suggests that such an extension should exist, but the
structure of the corresponding multi-terminal parameters remains to be
clarified.

Another direction is to make the relation with electric circuits more
explicit.  Quantum-walk scattering gadgets have been proposed as 
electric-circuit
realizations of universal quantum gates, using LC telegrapher circuits
\cite{ezawaElectricCircuitsUniversal2020}.  Related quantum walks with tails are known to
admit descriptions in terms of Kirchhoff-type laws and electric currents
\cite{higuchiElectricCircuitInduced2020}.  It would be interesting to clarify whether the
admittance-like quantities introduced here can be interpreted in a similar
circuit-theoretic framework, and whether the characteristic-polynomial
formulae obtained in this paper have a direct circuit interpretation.

It would also be interesting to apply the present scattering formalism to
topological quantum walks.  Topological phases of quantum walks have been
studied in one-dimensional systems with chiral symmetry
\cite{kitagawaExploringTopologicalPhases2010,asbothBulkboundaryCorrespondenceChiral2013}.  In such systems, topological
information is often encoded in the winding of a phase, and scattering-matrix
formulations of topological invariants are known both for topological
insulators and for discrete-time quantum walks
\cite{fulgaScatteringTheoryTopological2012,tarasinskiScatteringTheoryTopological2014}.  Since our formulae
express the two-terminal $S$-matrix in terms of characteristic polynomials,
they may give a finite-graph approach to scattering phases, winding numbers,
and possible edge-state contributions.

A related example is an Aharonov--Bohm ring.  Such a system can be viewed as
a two-terminal graph with parallel paths and magnetic phases in the edge
weights.  At the level of the path-sum formula, the magnetic flux enters
through the products of edge weights along paths, and therefore affects the
transmission amplitude and phase.  Extending the present real-weight
parametrization to this setting may lead to flux-controlled phase gates and
momentum filters.  We leave these problems for future work.

\section*{Acknowledgments}
%
This work was partially supported by JSPS KAKENHI Grant Numbers
JP26K17057, JP26K17055, and JP24K06889.

\printbibliography

@article{asbothBulkboundaryCorrespondenceChiral2013,
  title = {Bulk-Boundary Correspondence for Chiral Symmetric Quantum Walks},
  author = {Asb{\'o}th, J{\'a}nos K. and Obuse, Hideaki},
  year = 2013,
  month = sep,
  journal = {Phys. Rev. B},
  volume = {88},
  number = {12},
  pages = {121406},
  publisher = {American Physical Society},
  doi = {10.1103/PhysRevB.88.121406},
  urldate = {2026-05-03}
}

@article{childsLevinsonsTheoremGraphs2012,
  title = {Levinson's Theorem for Graphs {{II}}},
  author = {Childs, Andrew M. and Gosset, David},
  year = 2012,
  month = oct,
  journal = {Journal of Mathematical Physics},
  volume = {53},
  number = {10},
  eprint = {1203.6557},
  primaryclass = {math-ph, physics:quant-ph},
  pages = {102207},
  issn = {0022-2488, 1089-7658},
  doi = {10.1063/1.4757665},
  urldate = {2024-08-27},
  archiveprefix = {arXiv}
}

@article{childsMomentumSwitches2015,
  title = {Momentum {{Switches}}},
  author = {Childs, Andrew M. and Gosset, David and Nagaj, Daniel and Raha, Mouktik and Webb, Zak},
  year = 2015,
  journal = {Quantum Information and Computation},
  volume = {15},
  number = {7\&8},
  pages = {0601--0621},
  urldate = {2026-05-01}
}

@article{childsUniversalComputationMultiparticle2013,
  title = {Universal Computation by Multi-Particle Quantum Walk},
  author = {Childs, Andrew M. and Gosset, David and Webb, Zak},
  year = 2013,
  month = feb,
  journal = {Science},
  volume = {339},
  number = {6121},
  eprint = {1205.3782},
  primaryclass = {quant-ph},
  pages = {791--794},
  issn = {0036-8075, 1095-9203},
  doi = {10.1126/science.1229957},
  urldate = {2024-08-07},
  archiveprefix = {arXiv}
}

@article{childsUniversalComputationQuantum2009,
  title = {Universal {{Computation}} by {{Quantum Walk}}},
  author = {Childs, Andrew M.},
  year = 2009,
  month = may,
  journal = {Phys. Rev. Lett.},
  volume = {102},
  number = {18},
  pages = {180501},
  issn = {0031-9007, 1079-7114},
  doi = {10.1103/PhysRevLett.102.180501},
  urldate = {2024-08-07},
  copyright = {http://link.aps.org/licenses/aps-default-license},
  langid = {english}
}

@article{cottrellSimpleMethodFinding2015,
  title = {A Simple Method for Finding the Scattering Coefficients of Quantum Graphs},
  author = {Cottrell, Seth S.},
  year = 2015,
  month = sep,
  volume = {56},
  number = {9},
  pages = {092203--092203},
  issn = {0022-2488},
  doi = {10.1063/1.4931082},
  urldate = {2025-05-06},
  langid = {english}
}

@article{coutinhoIrrationalQuantumWalks2023,
  title = {Irrational {{Quantum Walks}}},
  author = {Coutinho, Gabriel and Baptista, Pedro Ferreira and Godsil, Chris and Spier, Thom{\'a}s Jung and Werner, Reinhard},
  year = 2023,
  month = sep,
  journal = {SIAM J. Appl. Algebra Geometry},
  volume = {7},
  number = {3},
  pages = {567--584},
  issn = {2470-6566},
  doi = {10.1137/22M1521262},
  urldate = {2026-05-05},
  langid = {english}
}

@article{coutinhoQuantumWalksNot2022,
  title = {Quantum Walks Do Not like Bridges},
  author = {Coutinho, Gabriel and Godsil, Chris and Juliano, Emanuel and {van Bommel}, Christopher M.},
  year = 2022,
  month = nov,
  journal = {Linear Algebra and its Applications},
  volume = {652},
  pages = {155--172},
  issn = {0024-3795},
  doi = {10.1016/j.laa.2022.07.009},
  urldate = {2025-04-24}
}

@article{ezawaElectricCircuitsUniversal2020,
  title = {Electric Circuits for Universal Quantum Gates and Quantum {{Fourier}} Transformation},
  author = {Ezawa, Motohiko},
  year = 2020,
  month = jun,
  journal = {Phys. Rev. Res.},
  volume = {2},
  number = {2},
  pages = {023278},
  publisher = {American Physical Society},
  doi = {10.1103/PhysRevResearch.2.023278},
  urldate = {2025-06-13}
}

@article{fulgaScatteringTheoryTopological2012,
  title = {Scattering Theory of Topological Insulators and Superconductors},
  author = {Fulga, I. C. and Hassler, F. and Akhmerov, A. R.},
  year = 2012,
  month = apr,
  journal = {Phys. Rev. B},
  volume = {85},
  number = {16},
  pages = {165409},
  issn = {1098-0121, 1550-235X},
  doi = {10.1103/PhysRevB.85.165409},
  urldate = {2026-05-03},
  copyright = {http://link.aps.org/licenses/aps-default-license},
  langid = {english}
}

@article{furerEfficientComputationCharacteristic2014,
  title = {Efficient {{Computation}} of the {{Characteristic Polynomial}} of a {{Tree}} and {{Related Tasks}}},
  author = {F{\"u}rer, Martin},
  year = 2014,
  month = mar,
  journal = {Algorithmica},
  volume = {68},
  number = {3},
  pages = {626--642},
  issn = {1432-0541},
  doi = {10.1007/s00453-012-9688-5},
  urldate = {2026-05-05},
  langid = {english}
}

@book{godsilAlgebraicCombinatorics1993,
  title = {Algebraic {{Combinatorics}}},
  author = {Godsil, C. D.},
  year = 1993,
  publisher = {Chapman \textbackslash\& Hall}
}

@article{godsilToolsLinearAlgebra1989,
  title = {Tools from Linear Algebra},
  author = {Godsil, Chris},
  year = 1989,
  month = jan,
  journal = {Handbook of Combinatorics},
  volume = {2}
}

@article{hararyDeterminantAdjacencyMatrix1962,
  title = {The {{Determinant}} of the {{Adjacency Matrix}} of a {{Graph}}},
  author = {Harary, Frank},
  year = 1962,
  month = jul,
  journal = {SIAM Rev.},
  volume = {4},
  number = {3},
  pages = {202--210},
  issn = {0036-1445, 1095-7200},
  doi = {10.1137/1004057},
  urldate = {2026-04-28},
  langid = {english}
}

@article{higuchiElectricCircuitInduced2020,
  title = {Electric {{Circuit Induced}} by {{Quantum Walk}}},
  author = {Higuchi, Yusuke and Sabri, Mohamed and Segawa, Etsuo},
  year = 2020,
  month = oct,
  journal = {J Stat Phys},
  volume = {181},
  number = {2},
  pages = {603--617},
  issn = {0022-4715, 1572-9613},
  doi = {10.1007/s10955-020-02591-3},
  urldate = {2025-06-10},
  langid = {english}
}

@article{kitagawaExploringTopologicalPhases2010,
  title = {Exploring Topological Phases with Quantum Walks},
  author = {Kitagawa, Takuya and Rudner, Mark S. and Berg, Erez and Demler, Eugene},
  year = 2010,
  month = sep,
  journal = {Phys. Rev. A},
  volume = {82},
  number = {3},
  pages = {033429},
  issn = {1050-2947, 1094-1622},
  doi = {10.1103/PhysRevA.82.033429},
  urldate = {2026-05-03},
  copyright = {http://link.aps.org/licenses/aps-default-license},
  langid = {english}
}

@book{pozarMicrowaveEngineering2012,
  title = {Microwave Engineering},
  author = {Pozar, David M.},
  year = 2012,
  edition = {4},
  publisher = {John Wiley \& Sons},
  urldate = {2026-05-05},
  isbn = {978-0-470-63155-3}
}

@article{sachsBeziehungenZwischenGraphen1964,
  title = {Beziehungen Zwischen Den in Einem {{Graphen}} Enthaltenen {{Kreisen}} Und Seinem Charakteristischen {{Polynom}}},
  author = {Sachs, Horst},
  year = 1964,
  journal = {Publicationes Mathematicae Debrecen},
  volume = {11},
  number = {1-4},
  pages = {119--134},
  publisher = {University of Debrecen/Debreceni Egyetem},
  urldate = {2026-04-28}
}

@inproceedings{schwenkComputingCharacteristicPolynomial1974,
  title = {Computing the {{Characteristic}} Polynomial of a Graph},
  booktitle = {Graphs and {{Combinatorics}}},
  author = {Schwenk, Allen J.},
  editor = {Bari, Ruth A. and Harary, Frank},
  year = 1974,
  pages = {153--172},
  publisher = {Springer},
  address = {Berlin, Heidelberg},
  doi = {10.1007/BFb0066438},
  isbn = {978-3-540-37809-9},
  langid = {english}
}

@inproceedings{schwenkRemovalcospectralSetsVertices1979,
  title = {Removal-Cospectral Sets of Vertices in a Graph},
  booktitle = {Proc. {{Tenth Southeastern Conference}} on {{Combinatorics}}, {{Graph Theory}} and {{Computing}}},
  author = {Schwenk, Allen J.},
  year = 1979,
  pages = {849--860}
}

@article{tarasinskiScatteringTheoryTopological2014,
  title = {Scattering Theory of Topological Phases in Discrete-Time Quantum Walks},
  author = {Tarasinski, B. and Asb{\'o}th, J. K. and Dahlhaus, J. P.},
  year = 2014,
  month = apr,
  journal = {Phys. Rev. A},
  volume = {89},
  number = {4},
  pages = {042327},
  issn = {1050-2947, 1094-1622},
  doi = {10.1103/PhysRevA.89.042327},
  urldate = {2026-05-05},
  copyright = {http://link.aps.org/licenses/aps-default-license},
  langid = {english}
}

@inproceedings{tutteAllKingsHorses,
  title = {All the King's Horses},
  booktitle = {Graph Theory and Related Topics},
  author = {Tutte, William Thomas},
  pages = {15--33},
  publisher = {Academic Press, Inc},
  address = {University of Waterloo, Ontario}
}

\pagebreak

\appendix

\section{Schwenk formula for weighted graphs} \label{schwenk}

Here we give a proof of the Schwenk recursion formula for the characteristic polynomial of a weighted graph with Hermitian weights. 
We give the proof in exactly the same way as the original proof by Schwenk \cite{schwenkComputingCharacteristicPolynomial1974}, but using the weighted version of the Harary-Sachs formula \cite{sachsBeziehungenZwischenGraphen1964} for the coefficients of the characteristic polynomial.
Let $D$ be a directed, weighted graph, with weight function $\omega : V \times V \to \mathbb{C}$. 
Later we will specialize to the case $\omega(u,v) = \omega(v,u)^*$.

The Harary-Sachs formula in the case of a weighted directed graph may be deduced from the decomposition of the coefficients of the characteristic polynomial into principal minors, and the Harary formula \cite{hararyDeterminantAdjacencyMatrix1962} for the determinant of an adjacency matrix. 
Let $S_r(D)$ denote the set of subgraphs of $D$ whose connected components are directed cycles.
For each directed graph $s \in S_r(D)$, let $p(s)$ denote the number of connected components of s.
Writing $\phi_D(y) = \sum_{r=0}^n a_r(D) y^{n-r}$ then we have 
\[
    a_r(D) = \sum_{s \in S_r(D)} (-1)^{p(s)}
    \prod_{C \subset s} \omega(C),
\]

Now decompose the sum by choosing a particular vertex $v$ and conditioning on whether $v$ is a member of $S_r$,
\[
    a_r(D) = \left(\sum_{s \in S_r(D\setminus v)}
    + \sum_{s \in S_r(D) | v \in s}\right)
     (-1)^{p(s)} 
    \prod_{C \subset s} \omega(C).
\]
The first sum is just $a_r(D\setminus v)$. 
For the other sum,
\begin{multline*}
    \sum_{s \in S_r(D) | v \in  s}
    (-1)^{p(s)}
    \prod_{C \subset s} \omega(C)
    \\=
    -\sum_{C\in \mathcal{C}(D) | v \in C}
    \omega(C)
    \sum_{s \in S_{r-|C|}(D\setminus C)}
    (-1)^{p(s)} 
    \prod_{C' \subset s} \omega(C')
    \\=
    -\sum_{C\in \mathcal{C}(D)  | v \in C}
    \omega(C) a_{r-|C|}(D\setminus C).
\end{multline*}
Putting the two formulae together, we have
\[
    a_r(D) = a_r(D\setminus v)
    -\sum_{C\in \mathcal{C}(D) | v \in C}
    \omega(C) a_{r-|C|}(D\setminus C).
\]

We now specialize to the case $\omega(v,u) = \omega(u,v)^*$, which we consider to be an undirected weighted graph $G$. 
The main difference is the choice of language; an undirected cycle in $G$, when considered as a directed graph, contains both the forward and reversed cycles, with weights $\omega(C)$ and $\omega(C)^*$.
Additionally, any edge $\{u,v\}$ in $G$ corresponds to a cycle of two vertices in the directed graph, with weight $|\omega(u,v)|^2$. 
This leads to the following decomposition:
\[
    a_r(G) = a_r(G\setminus v)
    -\sum_{u \sim v}
    |\omega(u,v)|^2 \, a_{r-2}(G\setminus u \setminus v)
    -2\sum_{C\in \mathcal{C}(G) | v \in C}
    \Re\left(\omega(C)\right) \, a_{r-|C|}(G\setminus C).
\]
Now multiply by $y^{n-r}$ and sum over $r$. 
Looking at the difference between the number of vertices of the graph and the subscript, we see that $a_{r}(G\setminus v)$ will pick up a factor of $y$, so
\begin{equation} \label{schwenk-hermitian}
    \phi_{G}(y) = y \phi_{G\setminus v}(y) 
    - \sum_{u \sim v}  |\omega(u,v)|^2 \phi_{G\setminus u \setminus v}(y)
    - 2\sum_{C\in \mathcal{C}(G) | v \in C}  \Re\left(\omega(C)\right) \phi_{G\setminus C}(y).
\end{equation}

\end{document}